\documentclass[authoryear]{elsarticle}
\usepackage{amsmath,amssymb,natbib}
\usepackage{graphicx,xcolor}

\def \R { \mathbb{R} }

\newcommand{\bsl}{\baselineskip}

\newcommand{\dd}{\mathrm{d}}

\newcommand{\bbt}{\boldsymbol{\beta}}

\newcommand{\sad}{{\varsigma_d}}
\newcommand{\rmin}{{r_{\min}}}
\newcommand{\mise}{\textsc{mise}}

\newcommand{\bb}{\mathbf{b}}
\newcommand{\br}{\mathbf{r}}

\newcommand{\bX}{\mathbf{X}}

\newcommand{\EE}{{\mathbb E}}

\usepackage{amsthm}
\usepackage{multirow}
\newtheorem{theorem}{Theorem}
\newtheorem{corollary}{Corollary}

\theoremstyle{definition}
\newtheorem{remark}{Remark}

\begin{document}
\begin{frontmatter}
\title{Second-order variational equations for spatial point processes with a view to pair correlation function estimation}
\author[label1]{Jean-Fran\c{c}ois Coeurjolly\corref{cor1}}
\address[label1]{Department of Mathematics, Universite\'e du Qu\'ebec \`a Montreal (UQAM), Canada}
\ead{coeurjolly.jean-francois@uqam.ca}

\author[label2]{Francisco Cuevas-Pacheco}
\ead{francisco@math.aau.dk}

\author[label2]{Rasmus Waagepetersen}
\address[label2]{Department of Mathematical Sciences, Aalborg University, Denmark}
\ead{rw@math.aau.dk}

\cortext[cor1]{Corresponding author}

\begin{abstract}
	\noindent Second-order variational type equations for spatial point processes are established. In case of log linear parametric models for pair correlation functions, it is demonstrated that the variational equations can be applied to construct estimating equations with closed form solutions for the parameter estimates. This result is used to fit orthogonal series expansions of log pair correlation functions of general form.
\end{abstract}


\end{frontmatter}


\section{Introduction}

Spatial point processes are models for sets of random locations of possibly interacting objects. Background on spatial point processes can be found in \cite{moeller:waagepetersen:04}, \cite{illian:etal:08} or \cite{baddeley:rubak:turner:15} which gives both  an accessible introduction as well as details on  implementation in the \texttt R package \texttt{spatstat}. Moments of counts of objects for spatial point processes are typically expressed in terms of so-called joint intensity functions or Papangelou conditional intensity functions which are defined via the Campbell or Georgii-Nguyen-Zessin equations (see the aforementioned references or the concise review of intensity functions and Campbell formulae in Section~\ref{sec:bck}). In this paper we consider a third type of equation called  variational equations. 

A key feature of variational equations compared to Campbell and Georgii-Nguyen-Zessin equations is that they are formulated in terms of the gradient of the log intensity or conditional intensity function rather than the (conditional) intensity itself. 
Variational equations were introduced for parameter estimation in Markov random fields by \cite{almeida:gidas:93}.
The authors suggested the terminology `variational' due to the analogy between the derivation of their estimating equation and the variational Euler-Lagrange equations in partial differential equations. The resulting equation consisted in an equilibrium equation involving the gradient of the log conditional probability of the Markov random field.
Later, \cite{baddeley:dereudre:13} obtained variational equations for Gibbs point processes and exploited them to infer a log-linear parametric model of the conditional intensity function. 
\cite{coeurjolly:moeller:14} established a first-order
variational equation for general spatial point processes and used it to estimate parameters in  a log-linear parametric model for the intensity function.

The first contribution of this paper is to establish
  second-order variational equations. The second-order properties of a
  spatial point process are characterized by the so-called pair
  correlation function which is a normalized version of the
  second-order joint intensity function. We assume that the pair
  correlation function is translation invariant and also consider the
  case when it is isotropic. Since the new variational equations are based on the gradient of the log pair correlation function, they take a particularly simple form for pair correlation functions of log-linear form.

Our second contribution is to propose a new non-parametric estimator of the pair correlation function. The classical approach is to use a kernel estimator, see e.g.~\cite{moeller:waagepetersen:04}. More recently, \cite{jalilian:guan:waagepetersen:17} investigated the estimation of the pair correlation function using an orthogonal series expansion. In the setting of their simulation studies, the orthogonal series estimator was shown to be more efficient than the standard kernel estimator. One drawback, however, is that the orthogonal series estimator is not guaranteed to be non-negative. We therefore propose to use our second-order variational equation to estimate coefficients in an orthogonal series expansion of the log pair correlation function. This ensures that the resulting  pair correlation function estimator is non-negative. We compare our new estimator with the previous ones in a simulation study and also illustrate its use on real datasets.

\section{Background and main results} \label{sec:bck}

\subsection{Spatial point processes}

Throughout this paper we let $\bX$ be a  spatial point process defined
on $\R^d$. That is, $\bX$ is a random subset of $\R^d$ with the
  property that the intersection of $\bX$ with any bounded subset of
  $\R^d$ is of finite cardinality. The joint intensity
  functions $\rho^{(k)}$, $k\ge 1$, are characterized (when they
  exist) by the Campbell formulae (equations)  \citep[see e.g.\ ][]{moeller:waagepetersen:04}:  for any $h:(\R^d)^k \to \R^+$ (with $\R^+$ the non-negative real numbers)
\begin{align}
\EE \sum_{u_1,\dots,u_k\in \bX}^{\neq} h(u_1,\dots,u_k)  &= \int \dots \int h(u_1,\dots,u_k)
\rho^{(k)}(u_1,\dots,u_k)
  \dd u_1 \dots \dd u_k.  \label{eq:campbell}
\end{align}
More intuitively, for any pairwise distinct points~$u_1,\dots,u_k \in \R^d$, \linebreak$\rho^{(k)}(u_1,\dots,u_k)\dd u_1 \cdots \dd u_k$ is the probability that for each~$i=1,\ldots,k$, $\bX$ has a point in an infinitesimally small region around~$u_i$ with volume~$\dd u_i$. The intensity function $\rho$ corresponds to the case $k=1$, i.e.\ $\rho=\rho^{(1)}$. The pair correlation function is obtained by normalizing the second-order joint intensity $\rho^{(2)}$:
  \begin{equation} \label{eq:pcf_def}
    g(u,v) = \frac{\rho^{(2)}(u,v)}{\rho(u)\rho(v)}
  \end{equation}
  for pairwise distinct $u,v$ and
  where $g(u,v)$ is set to 0 if~$\rho(u)$ or~$\rho(v)$ is
  zero. Intuitively, $g(u,v)>1$ [$g(u,v)<1$]  means that presence of a point at $u$ increases [decreases] the probability of observing a further point at $v$ and vice versa.  We assume that $\bX$ is observed on some bounded domain $W \subset \R^d$ with volume $|W|>0$ and without loss of generality we assume that $\rho(u)>0$ for all $u \in W$ (otherwise we just replace $W$ by $\{ u \in W | \rho(u)>0\}$ provided the latter set has positive volume).

We will always assume 
that $\bX$ is second-order intensity reweighted stationary
\citep{baddeley:moeller:waagepetersen:00}, meaning that its pair correlation
function $g$ is invariant by translations. We then,
with an abuse of notation,  write $g(v-u)$ for $g(u,v)$ for any $u,v\in
\R^d$. We will also consider the case of an isotropic pair correlation function in
which case $g(v-u)$ depends only on the distance $\|v-u\|$.

For the presentation of the second-order variational type
  equation in the next section some additional notation is needed.
For a function $h:\R^d \to \R$ which is differentiable on $\R^d$, we denote by 
\[
  \nabla h(w) = \left\{ \frac{\partial h}{\partial w_1}(w), \dots, \frac{\partial h}{\partial w_d}(w) \right\}^\top, \quad w\in \R^d
\]
the gradient vector with respect to the $d$ coordinates. The inner product is denoted by a `$\cdot$' and for $h: \R^d \to \R^d$, a multivariate function such that each component is differentiable on $\R^d$, we define the divergence operator
by
\[
\mathrm{div} \,h(w) = \sum_{i=1}^d \frac{\partial h_i}{\partial w_i}(w).
\]

\subsection{Second-order variational equations} \label{sec:VEstationary}


In this section, we present our new second-order variational
  equations. The proofs
  of the results are given in the Appendices.

\begin{theorem} \label{thm:VE}
Assume $\bX$ is second-order intensity reweighted stationary.
Let $h:\R^d \to \R^d$ be a componentwise continuously differentiable function on $\R^d$. Assume that $g$ is continuously differentiable on $\R^d$, that $\|h\| \|\nabla g\| \in L^1(\R^d)$, and that there exists a sequence of increasing bounded domains $(B_n)_{n\geq 1}$ such that $B_n \to \R^d$ as $n\to \infty$, with piecewise smooth boundary $\partial B_n$ and such that 
\begin{equation}
  \label{eq:condition_h}
\lim_{n\to \infty} \int_{\partial B_n} g(w) h(w) \cdot \nu(\dd w) =0
\end{equation}
where $\nu$ stands for the outer normal measure to $\partial B_n$. Then
\begin{align}
    \EE \Bigg\{ \sum_{u,v \in \bX\cap W }^{\neq} e(u,v)  \nabla \log g(v-u) \,& \cdot \,h(v-u)\Bigg\} =  \nonumber \\
   & - \EE \left\{ \sum_{u,v \in \bX \cap W }^{\neq} e(u,v)  \mathrm{div \, h}(v-u)
    \right\},\label{eq:VE1}
\end{align}
where $e:\R^d \times \R^d \to \R^+$ denotes the function $e(u,v) = \{\rho(u) \rho(v) |W \cap W_{v-u}|\}^{-1}$ for any $u,v \in \R^d$ and where $W_w$ denotes the domain $W$ translated by $w \in \R^d$.  
\end{theorem}
We note that condition~\eqref{eq:condition_h} is in particular satisfied if the function $h$ is compactly supported.

We next consider the case where the pair correlation function is isotropic, i.e. for any $u,v \in \R^d$ there exists $g_0:\R^+ \to \R^+$ such that $g(u,v)=g(v-u) = g_0(\|v-u\|)$.

\begin{theorem} \label{thm:VEisotropic}
Assume $\bX$ is second-order intensity reweighted stationary with
isotropic pair correlation function $g_0$.
Let $h:\R^+ \to \R$ be  continuously differentiable on $\R^+$.  Assume that $g_0$ is continuously differentiable on $\R^+$ and that either
\begin{align} 
  &t \mapsto h(t)g_0^\prime(t) \in L^1(\R^+) \quad \text{ and }
   \quad \lim_{n \to \infty } \{g_0(n)h(n)-g_0(0)h(0)\}=0 \label{eq:condition_h_iso1}
\end{align}
or
\begin{equation} \label{eq:condition_h_iso2}
   t \mapsto t^{d-1}h(t)g_0^\prime(t) \in L^1(\R^+) \; \text{ and } \; \lim_{n \to \infty } \{n^{d-1}g_0(n)h(n)-g_0(0)h(0)\mathbf 1(d=1)\}=0.
\end{equation}
Then we have the two following cases. If~\eqref{eq:condition_h_iso1} is assumed,
\begin{align}
    \EE \bigg\{ \sum_{u,v \in \bX\cap W }^{\neq} \frac{e(u,v)}{\|v-u\|^{d-1}} h(\|v-u\|) &(\log g_0)^\prime(\|v-u\|)\bigg\} = \nonumber\\
&    - \EE \left\{ \sum_{u,v \in \bX \cap W }^{\neq} \frac{e(u,v)}{\|v-u\|^{d-1}} h^\prime(\|v-u\|)
    \right\},\label{eq:VE_iso1}
\end{align}
where $e(u,v) = \{\rho(u) \rho(v) |W \cap W_{v-u}| \}^{-1}$ for any $u,v \in \R^d$. 
Instead, if \eqref{eq:condition_h_iso2} is assumed,
\begin{align}
    \EE \bigg\{ &\sum_{u,v \in \bX\cap W }^{\neq} e(u,v) h(\|v-u\|) (\log g_0)^\prime(\|v-u\|)\bigg\} = \nonumber\\
&    - \EE \left[ \sum_{u,v \in \bX \cap W }^{\neq} e(u,v) \left\{ (d-1)\frac{h(\|v-u\|)}{\|v-u\|}+h^\prime(\|v-u\|) \right\}
    \right].\label{eq:VE_iso2}
\end{align}
\end{theorem}

We stress  that the derivatives involved in Theorem~\ref{thm:VEisotropic} are derivatives with respect to $t\geq0$. Like for Theorem~\ref{thm:VE}, conditions~\eqref{eq:condition_h_iso1} and~\eqref{eq:condition_h_iso2} are in particular satisfied if $h$ is compactly supported in $(0,\infty)$.

\subsection{Sensitivity matrix}\label{sec:sensitivity}

In the next section we use empirical versions of \eqref{eq:VE_iso1} and \eqref{eq:VE_iso2} to construct
estimating functions for a parametric model of an isotropic pair correlation function $g_0$ depending on a $K$-dimensional parameter $\bbt$, $K \ge 1$. We here investigate the expression for the associated sensitivity matrices.

Consider functions $h_1,\ldots,h_K$ all fulfilling
\eqref{eq:condition_h_iso1} and possibly depending on $\bbt$. By
stacking the $K$ equations obtained by applying these functions for
$h_1,\dots,h_K$ in \eqref{eq:VE_iso1} we obtain the estimating function
\begin{equation}
      \sum_{u,v \in \bX\cap W }^{\neq} \frac{e(u,v)}{\|v-u\|^{d-1}} \mathbf{h}(\|v-u\|) (\log g_0)^\prime(\|v-u\|) + \sum_{u,v \in \bX \cap W }^{\neq} \frac{e(u,v)}{\|v-u\|^{d-1}} \mathbf{h}^\prime(\|v-u\|) \label{eq:estmfct}
    \end{equation}
where $\mathbf{h}$ and $\mathbf{h}'$ are vector functions with
components $h_i$ and $h_i'$. The sensitivity matrix is obtained as the
expectation of the negated derivative (with respect to $\bbt$) of
\eqref{eq:estmfct}. After applying \eqref{eq:VE_iso1} once again
after differentiation we obtain the sensitivity matrix
\[ S(\bbt)=- \EE \sum_{u,v \in \bX\cap W }^{\neq} \frac{e(u,v)}{\|v-u\|^{d-1}} \mathbf{h}(\|v-u\|) \frac{\dd}{\dd \bbt^\top}  (\log g_0)^\prime(\|v-u\|). \]
Applying the Campbell theorem and converting to polar coordinates, we obtain
\[ S(\bbt) = - \sad \int_0^\infty \mathbf{h}(t)\left
    [\frac{\dd}{\dd \bbt^\top}  (\log g_0)^\prime(t) \right] g_0(t)
  \dd t, \]
where $\sad$ is the surface area of the $d$-dimensional unit ball.
In case of \eqref{eq:VE_iso2} we obtain a similar expression,
\[ S(\bbt)=-  \sad \int_0^\infty \mathbf{h}(t)\left [\frac{\dd}{\dd \bbt^\top}  (\log g_0)^\prime(t) \right] g_0(t) t^{d-1} \dd t.  \]
By choosing $\mathbf{h}(t)=-\psi(t) \frac{\dd}{\dd \bbt}  (\log g_0)^\prime(t)$ for some real function $\psi$, $S(\bbt)$ becomes at least positive semi-definite.

\section{Estimation of log linear pair correlation function} \label{sec:loglin}

We now consider the estimation of an isotropic pair correlation function of the form
\begin{equation}\label{eq:g_isotropic}
\log g_0(t) = \bbt^\top \br(t) = \bbt^\top \left\{ r_1(t),\dots,r_K(t)\right\}^\top  
\end{equation}
where the functions $r_k:\R^+\to \R$, $k=1,\dots,K$ are known. Following Section~\ref{sec:sensitivity}, the idea is to apply Theorem~\ref{thm:VEisotropic} $K$ times to  functions $h_i$, $i=1,\ldots,K$, of the form $h_i(t) = {-\psi(t) }\frac{\partial}{\partial \bbt_i} (\log g_o)'(t) = -\psi(t) r^\prime_i(t)$ where the function $\psi:\R^+ \to \R$ will be justified and specified  later. It is then remarkable that we obtain a simple estimating equation of the form $\mathbf{A} \bbt + \mathbf b =0$.
The sensitivity matrix discussed in Section~\ref{sec:sensitivity} is $S(\bbt)=-\EE \mathbf A$. Provided $\mathbf A$ is invertible we obtain the explicit solution
\begin{equation} \label{eq:betaEst}
  \hat \bbt = - \mathbf A^{-1} \mathbf b.
\end{equation}
The matrix $\mathbf A$ and the vector $\mathbf b$ are specified in the following corollary.
\begin{corollary}\label{cor:VEisotropic}
Let $\psi:\R^+\to \R$. Assume that $\psi$ and $r_k$ ($k=1,\dots,K$) are respectively  continuously differentiable and twice  continuously differentiable on $\R^+$. 
Assume either that
\begin{equation}
  \label{eq:rellphi1}
  t \mapsto \|\br^\prime(t)\|^2 \psi(t) \in L^1(\R^d) \text{ and }
  \lim_{n\to \infty} \psi(n) \br(n)^\top \br^\prime(n) - \psi(0) \br(0)^\top\br^\prime(0)=0
\end{equation}
or
\begin{align}
  & t \mapsto t^{d-1} \|\br^\prime(t)\|^2 \psi(t) \in L^1(\R^d) \nonumber\\
&\text{ and }
  \lim_{n\to \infty} n^{d-1}\psi(n) \br(n)^\top \br^\prime(n) - \psi(0) \br(0)^\top\br^\prime(0)\mathbf 1(d=1)=0.
  \label{eq:rellphi2}
\end{align}
If~\eqref{eq:rellphi1} is assumed, we define the $(K,K)$ matrix $\mathbf A$ and the vector $\mathbf b\in \R^K$ by
\begin{align} \
  \mathbf A&= \sum_{u,v \in \bX \cap W}^{\neq} \frac{e(u,v)}{\|v-u\|^{d-1}}\psi(\|v-u\|) \br^\prime(\|v-u\|)\{\br^\prime(\|v-u\|)\}^\top  \label{eq:defA1}\\
\bb&=\sum_{u,v \in \bX \cap W }^{\neq} \frac{e(u,v)}{\|v-u\|^{d-1}}  \left\{ \psi^\prime(\|v-u\|) \br^\prime(\|v-u\|) + \psi(\|v-u\|) \br^{\prime\prime}(\|v-u\|)\right\}  \label{eq:defb1}
\end{align}
where again the edge effect factor is  $e(u,v) = \{\rho(u) \rho(v) |W \cap W_{v-u}| \}^{-1}$ for any $u,v\in \R^d$.
Instead, in case of \eqref{eq:rellphi2}, we define
\begin{align} \
  \mathbf A&= \sum_{u,v \in \bX \cap W}^{\neq} e(u,v)\psi(\|v-u\|) \br^\prime(\|v-u\|)\{\br^\prime(\|v-u\|)\}^\top  \label{eq:defA2}\\
\bb&=\sum_{u,v \in \bX \cap W }^{\neq} e(u,v)  \bigg\{ 
(d-1) \frac{\psi(\|v-u\|) \br^\prime(\|v-u\|)}{\|v-u\|} \nonumber \\
& \qquad \qquad
+ \psi^\prime(\|v-u\|) \br^\prime(\|v-u\|) + \psi(\|v-u\|) \br^{\prime\prime}(\|v-u\|)\bigg\}  \label{eq:defb2}
\end{align}
Then, the equation 
\begin{equation}
  \label{eq:VEexp} \mathbf A \bbt + \bb =0
\end{equation} 
is an unbiased estimating equation. 
\end{corollary}

\begin{proof} The proof consists in applying Theorem~\ref{thm:VEisotropic} with $h(t)=-\psi(t) r_k'(t)$ for $k=1,\dots,K$ and in noticing that $(\log g_0)^\prime(t) = \bbt^\top \br^\prime(t) = \br^\prime(t)^\top \bbt$.
  \end{proof}

We note that if $\psi$ is compactly supported in $[0,\infty)$,
then~\eqref{eq:rellphi1} or~\eqref{eq:rellphi2} are always valid
assumptions. Another special case is also interesting: let $d>1$ and
$\psi=1$, then~\eqref{eq:rellphi2} is true if for any $k,l=1,\dots,K$,
$t\mapsto t^{d-1}r^\prime_k(t)^2 \in L^1(\R^d)$ and $\lim_{n\to
  \infty} n^{d-1}r_k(n)r_{l}^\prime(n)=0$. This simple
  condition is for instance satisfied if the $r_k$'s' are exponential covariance functions. 

The results above are e.g.\ applicable to the case of a pair correlation function for a log Gaussian Cox process with covariance function given by a sum of known correlation functions scaled by unknown variance parameters. Assuming a known correlation function is on the other hand quite restrictive. However, any log pair correlation function can be approximated well on a finite interval using a suitable basis function expansion so that we can effectively represent it as a log linear model. We exploit this in Section~\ref{sec:osenonpar} where we consider the case where the functions $r_k$ are basis functions on a bounded real interval.

\begin{remark}
In the more general case of a translation invariant pair correlation function of log linear form
\begin{equation} \label{eq:g}
  \log g(w) = \sum_{k=1}^K \bbt_k r_k(w) = \bbt^\top \br(w)
\end{equation}
where $r_k(w)$, $k=1,\dots,K$, are $K$ known functions assumed to be continuously differentiable on $\R^d$, we can also obtain an estimating equation of the form \eqref{eq:VEexp} using Theorem~\ref{thm:VE} instead of Theorem~\ref{thm:VEisotropic}. We omit the details.
\end{remark}
\begin{remark}\label{rm:psi}
In applications of \eqref{eq:defA1}-\eqref{eq:defb1} for $d=2$ or  \eqref{eq:defA2}-\eqref{eq:defb2} for $d \ge 1$ the division by
  $\|v-u\|^{d-1}$ or $\|v-u\|$ may lead to numerical instability for
  pairs of close points $u$ and $v$. This can be mitigated by
  a proper choice of the function $\psi$. In the spatial case of $d=2$ we propose to define
  $\psi(t)=(t/b)^2(1-(t/b))^2 \mathbf 1(t \in[0,b])$ for some $b>0$. With this
  choice of $\psi$ the divisors $\|v-u\|^{d-1}=\|v-u\|$ cancel out preventing very large or infinite variances of
 \eqref{eq:defA1}-\eqref{eq:defb2}. 
\end{remark}

\begin{remark}
The quantities \eqref{eq:defA1}-\eqref{eq:defb2} depend on the unknown intensity function. If the intensity function is constant equal to $\rho>0$ we can multiply \eqref{eq:VEexp} by $\rho^2$ whereby the resulting estimating equation no longer depends on $\rho$. Thus $g_0$ can be estimated without estimating $\rho$.  Otherwise, the intensity function has to be estimated first, for instance in a parametric way, see~\citet{guan:jalilian:waagepetersen:15}, and plugged into~\eqref{eq:defA1}-\eqref{eq:defb2}.
\end{remark}


\section{Variational orthogonal series estimation of the pair correlation function}\label{sec:osenonpar}

In this section we consider the estimation of an isotropic pair correlation function $g_0$ on a bounded interval $[r_{\min},r_{\min}+R]$, $0\le r_{\min}<\infty$ and $0<R<\infty$, using a series expansion of $\log g_0$. Let $\{\phi_k\}_{k\geq 1}$ denote an orthonormal basis of functions on $[0,R]$ with respect to some weight function $w(\cdot)\ge 0$, i.e. $\int_{0}^R  \phi_{k}(t)\phi_{l}(t) w(t) \dd t= \delta_{k l}$. Provided $\log g_0$ is square integrable (with respect to $w(\cdot)$) on $[r_{\min},r_{\min}+R]$, we have the expansion
\begin{equation}
  \label{eq:logg_orthogonal}
  \log g_0(t) = \sum_{k=1}^\infty \beta_k \phi_k(t-r_{\min})
\end{equation}
where the coefficients $\beta_k$ are defined by $\beta_k = \int_0^R g_0(t+r_{\min})\phi_k(t) w(t)\dd t$. 

We propose to approximate $\log g_0$ by truncating the infinite sum up
to some $K \ge 1$ and obtain estimates $\hat \beta_1,\ldots,\hat \beta_K$ using \eqref{eq:VEexp}. The resulting estimate thus becomes
\[ \widehat{\log g_{0,K}}(t) = \sum_{k=1}^K \hat \beta_k \phi_k(t-r_{\min}) .\]
In the sequel this estimator is referred to as the variational (orthogonal series) estimator (VSE for short).
The approach is related to \cite{zhao:18} who also considers an estimating equation approach to estimate a pair correlation function of the form \eqref{eq:logg_orthogonal} but for a number $m>1$ of independent point processes on $\R$. The approach in \cite{zhao:18} further does not yield closed form expressions for the estimates of the coefficients. 

Orthogonal series estimators have already been considered by
\citet{jalilian:guan:waagepetersen:17} who expand $g_0-1$ instead of $\log g_0$. They propose very simple unbiased estimators of the coefficients but the resulting estimator of $g_0$, referred to as the OSE in the sequel, is not guaranteed to be non-negative. 

\subsection{Implementation of the VSE}

Examples of orthogonal bases include the cosine basis with $w(r)=1$, $\phi_1(r) = 1/\surd{R}$ and
$\phi_k(r)= (2/R)^{1/2} \cos\{ (k - 1) \pi r/R \}$, $k \ge 2$. Another
example is the Fourier-Bessel basis with $w(r)= r^{d-1}$ and
 \[ 
   \phi_k(r)=\frac{2^{1/2}}{ RJ_{\nu+1}(\alpha_{\nu,k})} J_{\nu}\left(r \alpha_{\nu,k}/R\right)r^{-\nu}, \quad k \ge 1,
 \]
where $\nu=(d-2)/2$, $J_{\nu}$ is the Bessel function of the first kind of 
order $\nu$ and $\{\alpha_{\nu,k}\}_{k=1}^\infty$ is the sequence of
successive positive roots of $J_{\nu}(r)$. In the context of the variational equation \eqref{eq:VEexp} we need that the basis functions $\phi_k$ have non-zero derivatives in order to estimate $\beta_k$. This is not the case for $\phi_1$ of the cosine basis. We therefore consider in the following the Fourier-Bessel basis.

Let $b_k=1[k \le K]$, $k \ge 1$. The mean integrated squared
error (MISE) for $\log g_0$ of the VSE over the interval $[\rmin,R+\rmin]$ is
\begin{align}
  \mise\big(\widehat{\log g_{0,K}}\big) 
  &= \sad \int_{\rmin}^{\rmin+R}  
  \EE\big\{ \widehat{\log {g_{0,K}}}(r) - \log {g_{0,K}}(r) \big\}^2 w(r-\rmin) \dd
    r \label{eq:miseo} \\
  &= \sad \sum_{k=1}^{\infty} \EE(b_k \hat{\beta}_{k}  - \beta_k)^2 
  =  \sad \sum_{k=1}^{\infty} \big[ b_{k}^2 \EE\{\hat{\beta}_{k}^2\} 
  -2b_k \beta_{k} \EE \hat \beta_{k}+ \beta_{k}^2 \big]. \nonumber
\end{align} 
\cite{jalilian:guan:waagepetersen:17} chose $K$ by minimizing an estimate of the MISE for $g_0$. We have, however, not been able to construct a useful estimate of \eqref{eq:miseo}. Instead we choose $K$ by maximizing a composite likelihood cross-validation criterion 
\begin{multline*} \mathrm{CV}(K)=\sum_{\substack{u, v \in
      \bX \cap W:\\r_{\min} \le \|u-v\|\le r_{\min}+\R}}^{\neq} \!\!\!\!\log[ \rho(u)\rho(v) \exp[\widehat{\log g_{0,K}}^{-\{u,v\}} ({\|v-u\|})]\\
  - \!\!\! \!\!\!\!\! \sum_{{\substack{u, v \in
      \bX \cap W:\\r_{\min} \le \|u-v\|\le r_{\min}+\R}}}^{\neq}
\!\!\!\!\!  \!\!\!\!\!\!\!\! \log
\int_{{W^2}} \!\!\!\! 1[ r_{\min} \le \|u-v\| \le r_{\min}+R] \rho(u)\rho(v) \exp[ \widehat{\log g_{0,K}}({\|v-u\|})] \dd u \dd v \end{multline*}
where $\widehat{\log g_{0,K}}^{-\{u,v\}}$ is the estimate of $\log g_0$ obtained using all pairs of points in $\bX$ except $(u,v)$ and $(v,u)$. This is a simplified version of the cross-validation criterion introduced by \cite{guan:composite:07} in the context of non-parametric kernel estimation of the pair correlation function. 

For computational simplicity and to guard against overfitting we choose inspired by \cite{jalilian:guan:waagepetersen:17} the first local maximum of $\mathrm{CV}(K)$ larger than or equal to two rather than looking for a global maximum. Note that when $\mathbf A$ and $\mathbf b$ in \eqref{eq:VEexp} have been obtained for one value of $K$, then we obtain the $\mathbf A$ and $\mathbf b$ for $K+1$ by just adding  one new row/column to the previous $\mathbf A$ and one new entry to the previous $\mathbf b$.

\subsection{Simulation study}\label{sec:simstudy}

We study the performance of our variational estimator using simulations of point processes with constant intensity 200 on $W=[0,1]^2$ or $W=[0,2]^2$. We consider the case of a Poisson process for which the pair correlation function is constant equal to one, a Thomas process (parent intensity $\kappa = 25$, dispersal standard deviation $\omega = 0.0198$ and offspring intensity $\mu = 8$), a variance Gamma cluster process (parent intensity $\kappa = 25$,  shape parameter $\nu = -1/4$,  dispersion parameter $\omega =  0.01845$ and offspring intensity $\mu = 8$), and a determinantal point process (DPP) with exponential kernel $K(r)=\exp(-r/\alpha)$ and $\alpha=0.039$. The pair correlation functions for the four point process models are shown in Figures~\ref{fig:estimates1} and~\ref{fig:estimates2} in the usual scale as well as in the log scale. The Thomas and variance Gamma processes are clustered with pair correlation functions bigger than one while the DPP is repulsive with pair correlation function less than one. 
In all cases we consider $R=0.125$ and we let $r_{\min}=0$ for Poisson, Thomas, and variance Gamma. For the DPP the log pair correlation function is not well-defined for $r=0$ and we therefore use $r_{\min}=0.01$ in case of the DPP. We use \eqref{eq:defA1} and \eqref{eq:defb1} for computing $\mathbf A$ and $\mathbf b$ and referring to Remark~\ref{rm:psi} we let $b=r_{\min}+R$. For each point process we generate 500 simulations.

\subsubsection{Estimates of coefficients}

Equations \eqref{eq:defA1} and \eqref{eq:defb1} are derived from
\eqref{eq:VE_iso1} in which $g_0$ is the true pair correlation
function. In practice, when considering a truncated version of
\eqref{eq:logg_orthogonal}, the estimating equation~\eqref{eq:VEexp}
is not unbiased which results in bias of the coefficient
estimates. This is exemplified in case of the Thomas process in the left plot of
Figure~\ref{fig:coefest} which shows boxplots of the first two
coefficient estimates when~\eqref{eq:logg_orthogonal} is truncated to
$K=2$. In the right plot, \eqref{eq:logg_orthogonal} is truncated to
$K=8$ which means that the truncated version
of~\eqref{eq:logg_orthogonal} is very close to the Thomas pair correlation function. Accordingly,
the bias of the estimates is much reduced. However, the estimation
variance increases when $K$ is increased. This emphasizes the
importance of selecting an appropriate trade-off between bias an
variance. The plots in Figure~\ref{fig:coefest} also show how
  the variance of the coefficient estimates decreases when the
  observation window $W$ is increased from $[0,1]^2$ to $[0,2]^2$.
\begin{figure}
  \centering
\begin{tabular}{cc}
\includegraphics[width=0.48\linewidth]{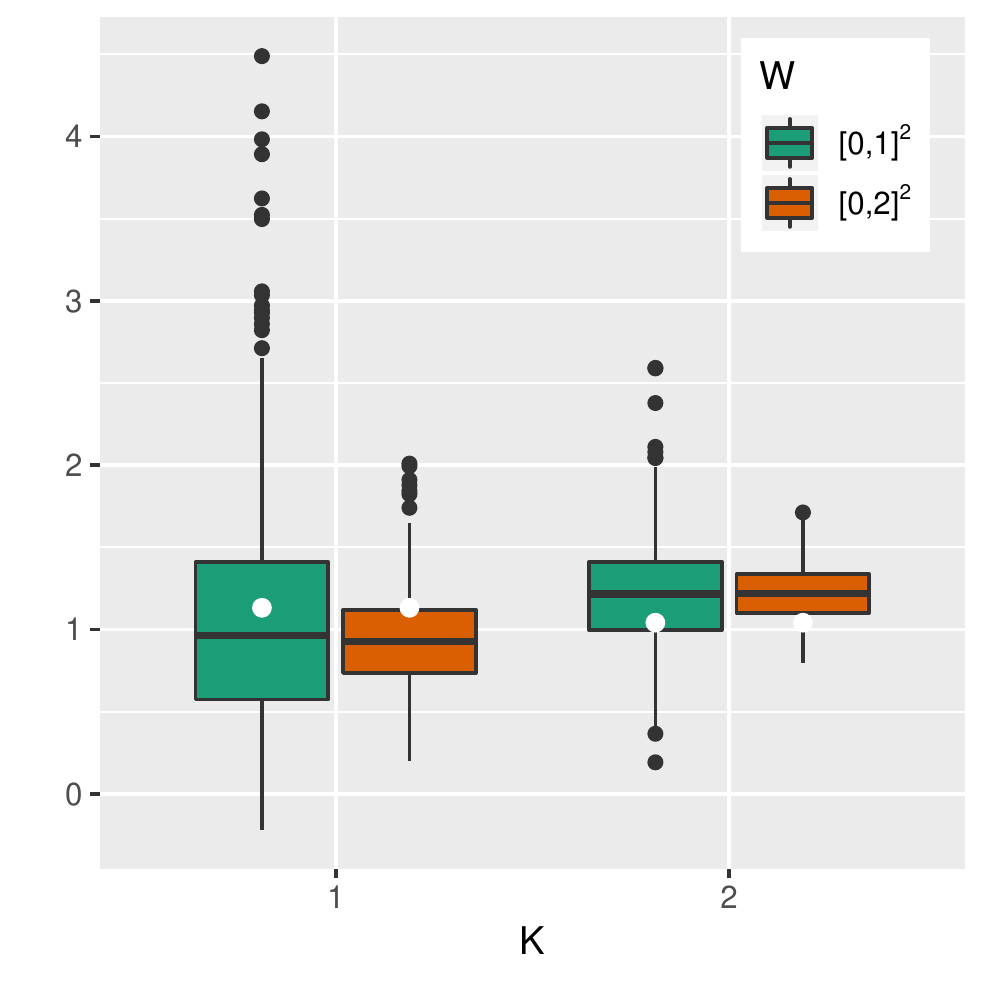} &  \includegraphics[width=0.48\linewidth]{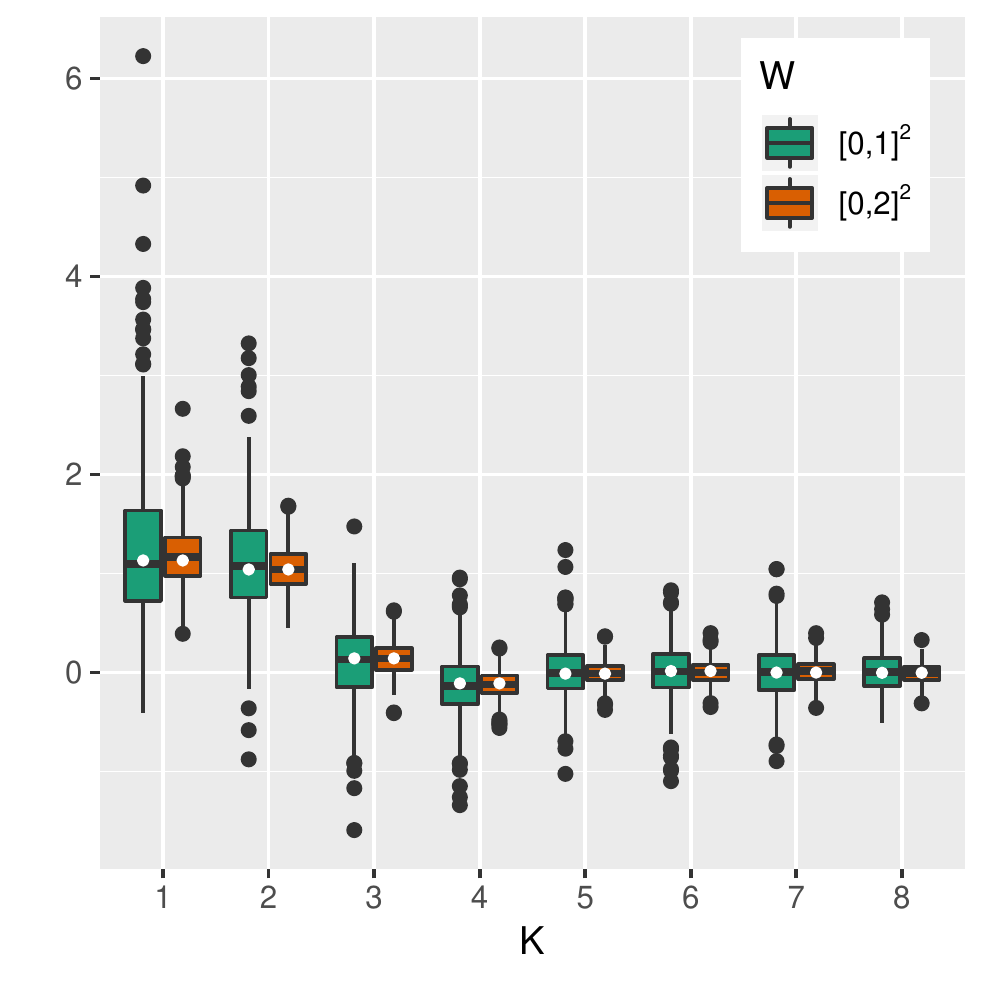}
\end{tabular}
  \caption{Estimates of the first $K$ coefficients when
    \eqref{eq:logg_orthogonal} is truncated to $K=2$ (left) or $K=8$
    (right) in case of the Thomas process. White points
    correspond to the true coefficient values. Observation
    window is either $W=[0,1]^2$ or $W=[0,2]^2$.}\label{fig:coefest}
\end{figure}

\subsubsection{Comparison of estimators}

In addition to our new VSE, we also for each simulation consider the OSE proposed by \cite{jalilian:guan:waagepetersen:17} (using the
Fourier-Bessel basis and their so-called simple smoothing scheme) and
a standard non-parametric kernel density estimate (KDE) with bandwidth
chosen by cross-validation
\citep{guan:leastsq:07,jalilian:waagepetersen:18}. Figures~\ref{fig:estimates1}
and \ref{fig:estimates2} depict means of the simulated OSE and
  VSE estimates of $g_0$ and $\log g_0$ as well as 95\% pointwise envelopes.
The variational estimator has larger variability whereas the bias can be smaller or larger than the OSE depending on the model.

\begin{figure}[t!] 
\centering
\begin{tabular}{cc}
\includegraphics[scale=.55]{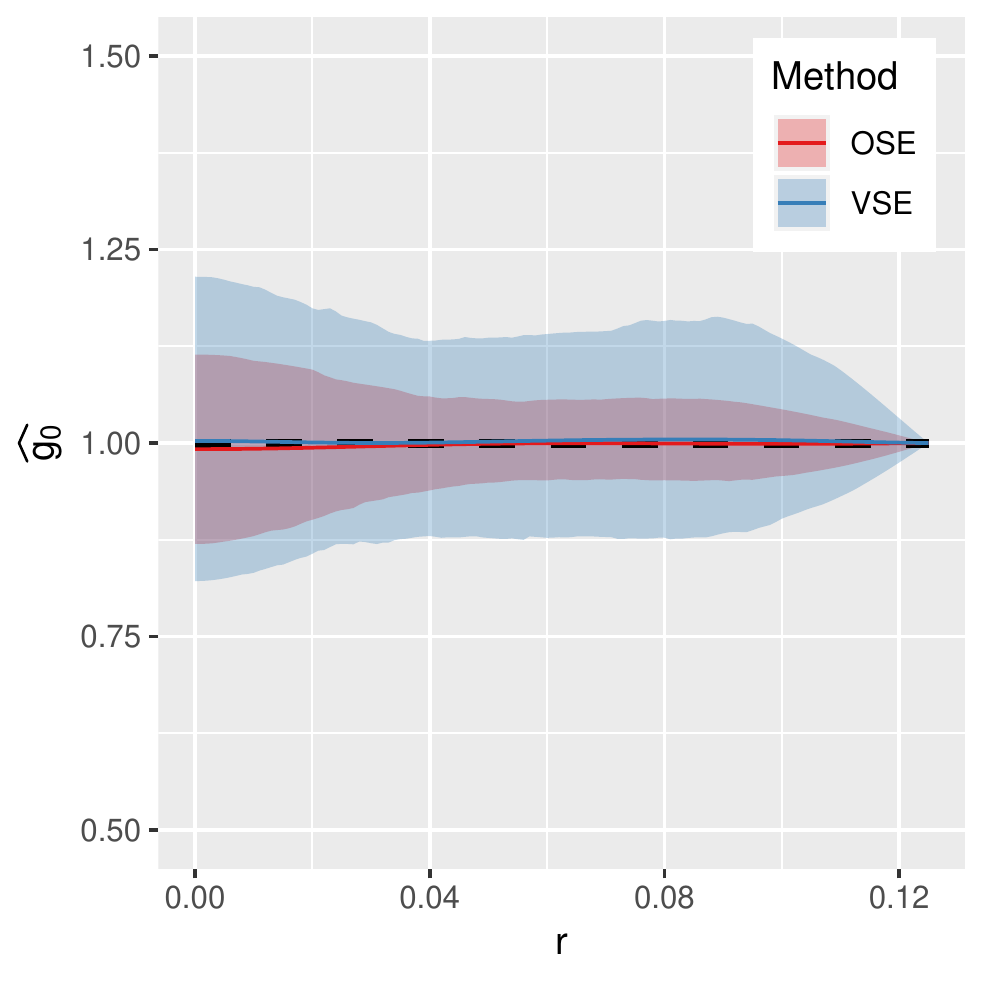}&
\includegraphics[scale=.55]{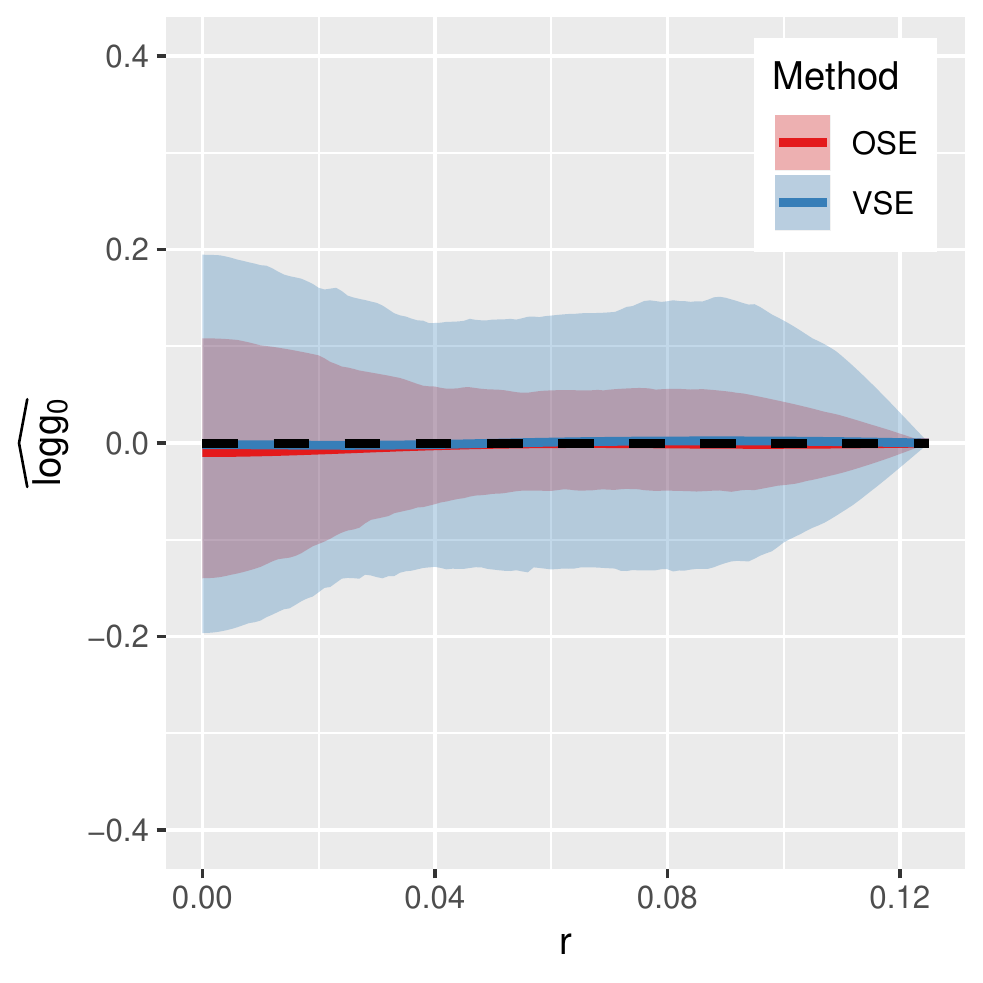}\\
\includegraphics[scale=.55]{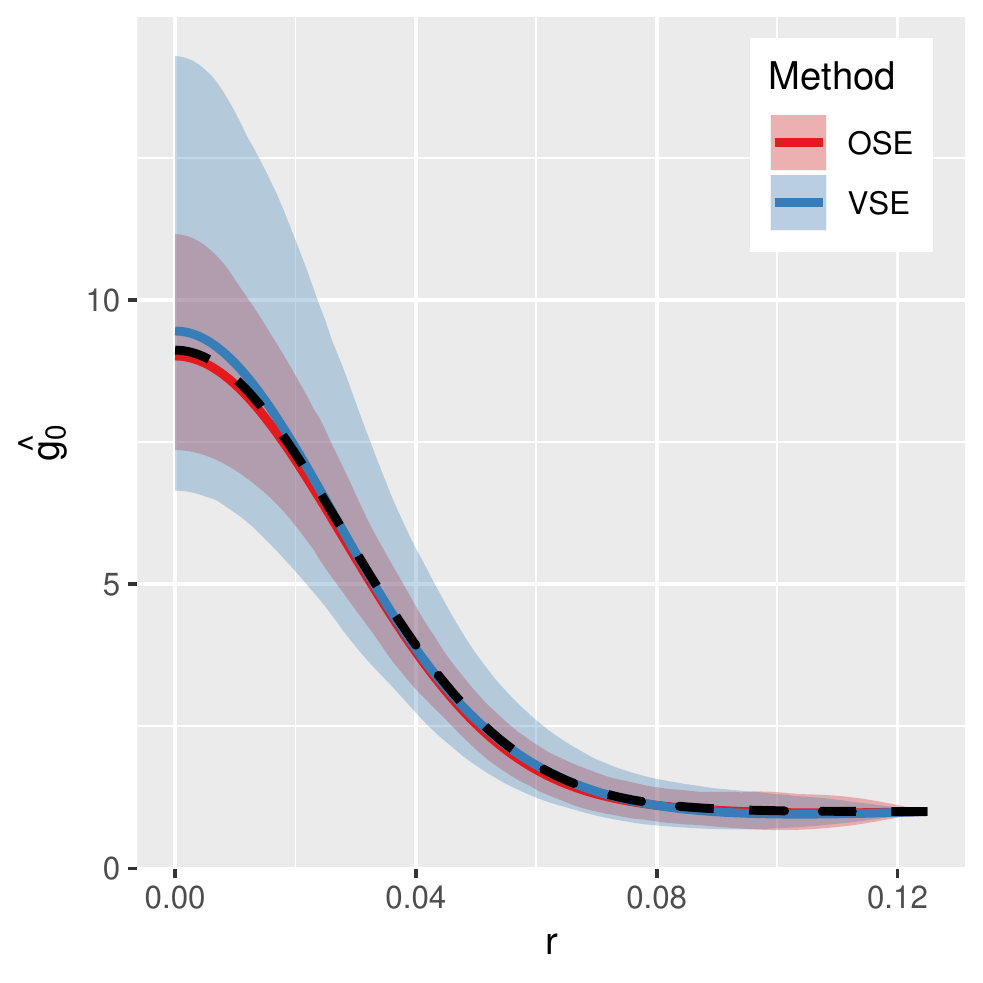}&
\includegraphics[scale=.55]{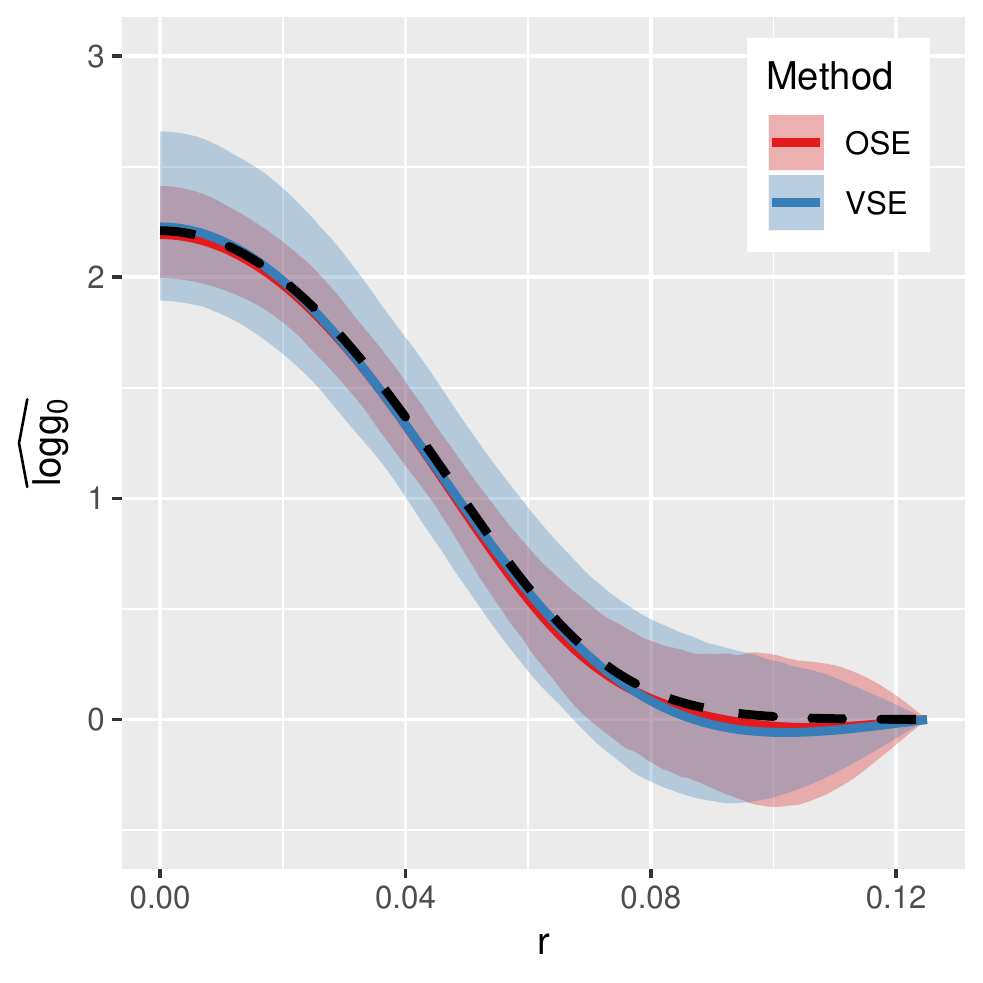}
\end{tabular}
\caption{Mean VSE (red curves) and OSE (blue curves) of $g_0$ (first column) and $\log g_0$
  (right column) for Poisson (first row) and Thomas
  (second row) point
processes with $W=[0,2]^2$. In each plot, the dashed black curve is
the true pair correlation  or log pair correlation function. The
envelopes represent pointwise 95\% probability intervals for the estimates.}\label{fig:estimates1}
\end{figure}

\begin{figure}[t!] 
\centering
\begin{tabular}{cc}
\includegraphics[scale=.55]{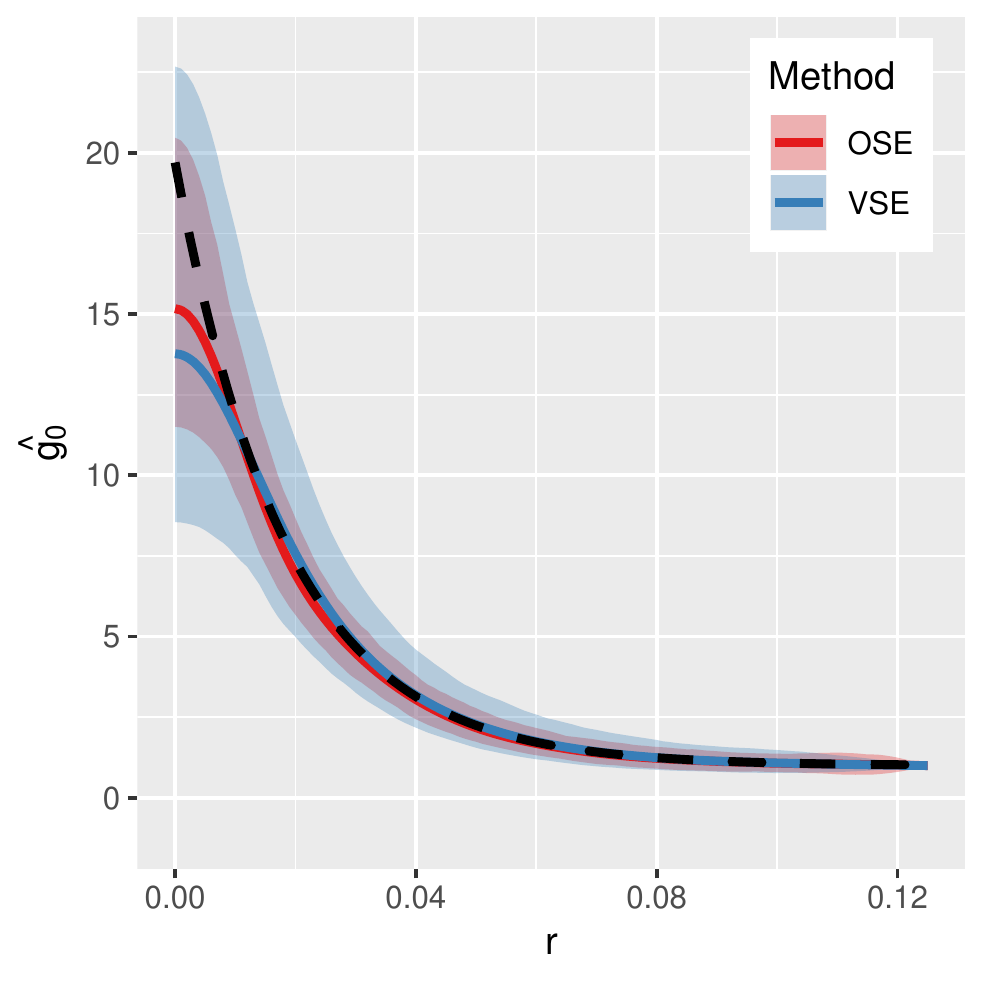}&
\includegraphics[scale=.55]{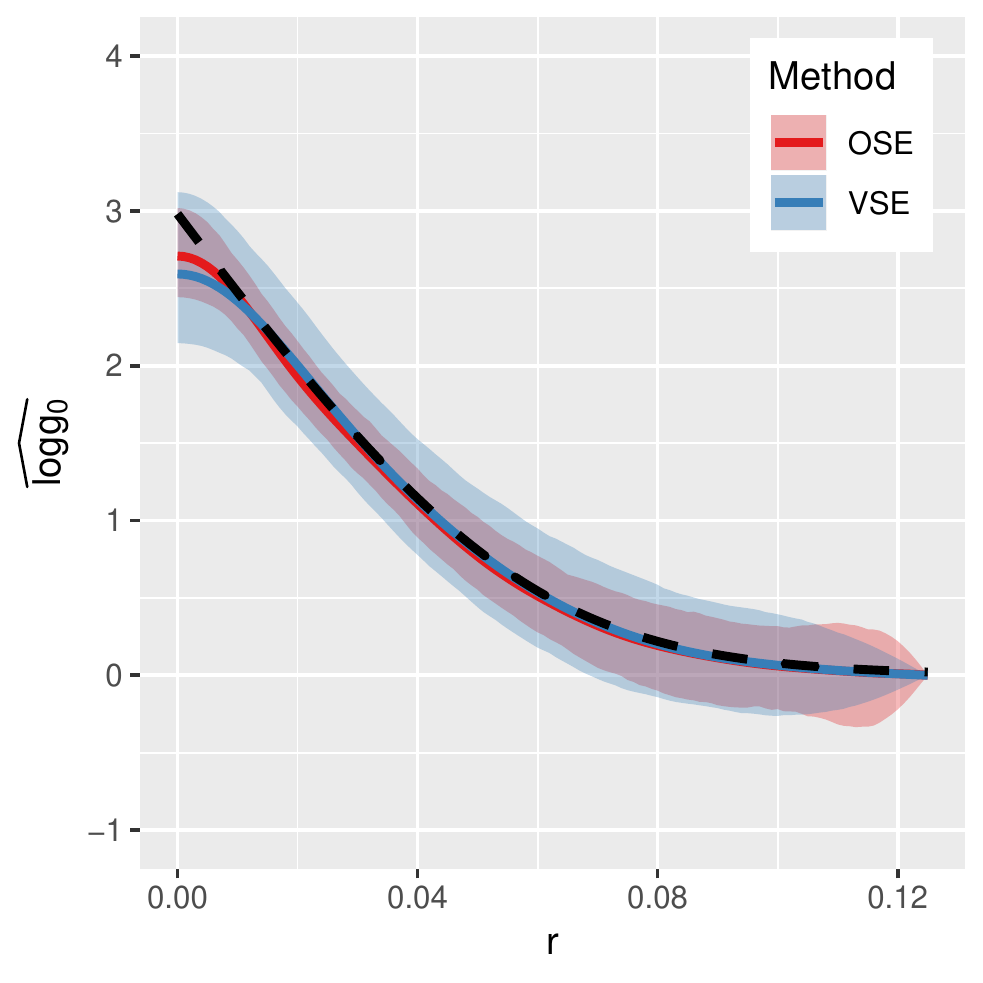}\\
\includegraphics[scale=.55]{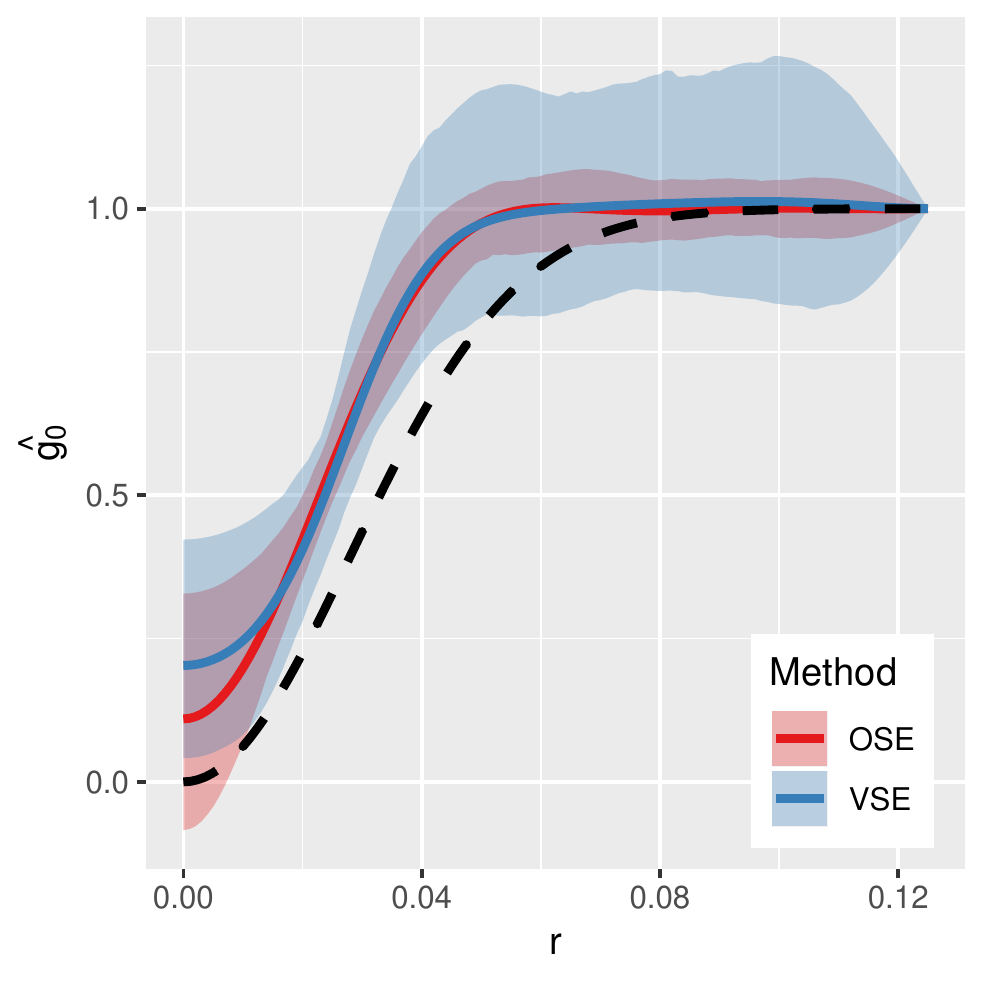}&
\includegraphics[scale=.55]{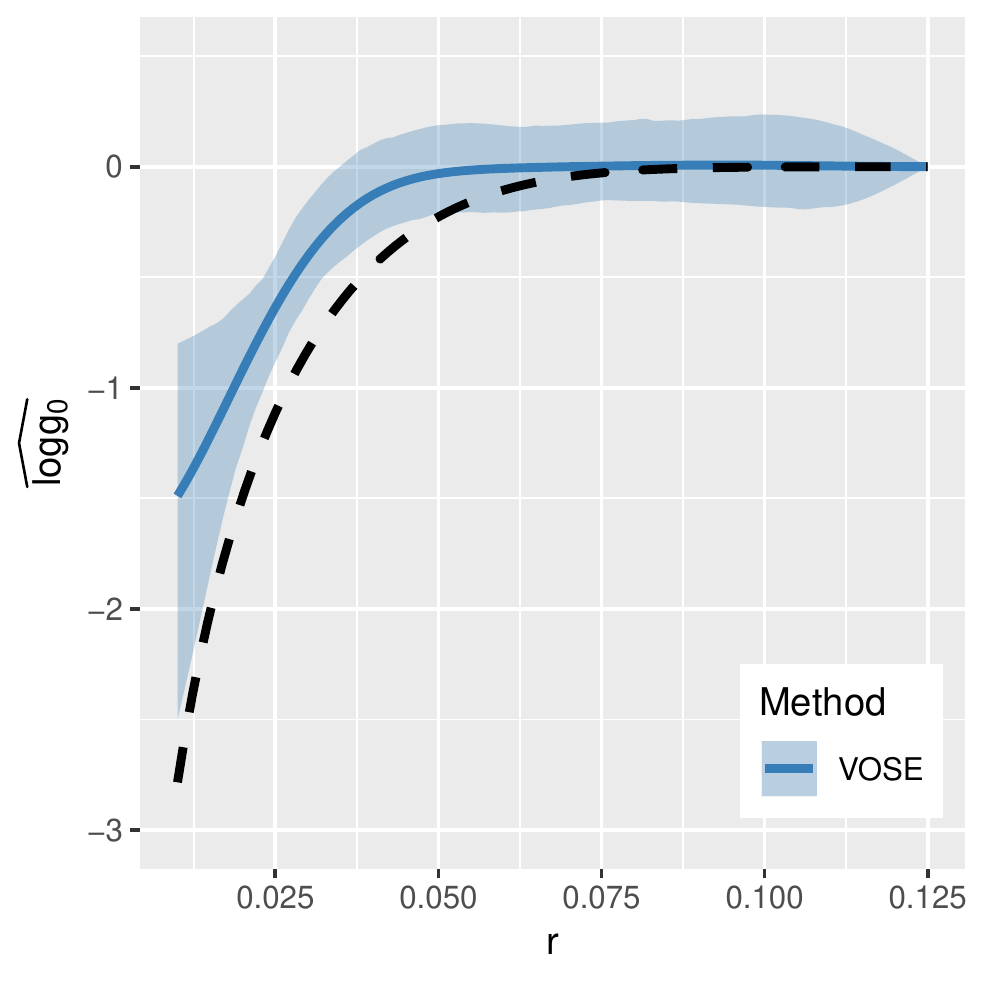}
\end{tabular}
\caption{
Mean VSE (red curves) and OSE (blue curves) of $g_0$ (first column) and $\log g_0$
  (right column) for variance gamma (first row) and determinantal
  (second row, $r_{\min}=0.01$) point
processes with $W=[0,2]^2$. In each plot, the dashed black curve is
the true pair correlation  or log pair correlation function. The
envelopes represent pointwise 95\% probability intervals for the
estimates.}\label{fig:estimates2}
\end{figure}

Table~\ref{tbl:mise} summarizes the root MISE (square root of \eqref{eq:miseo}) for the three estimators across the four models. 
\begin{table}[!htb]
\begin{center}
\begin{tabular}{lllll}
                        & Window $\quad$   &  OSE $\quad \quad$&  VSE $\quad \quad$ & KDE $\quad \quad$ \\ 
                        \hline
Poisson                     & $[0,1]^2$ &  0.027   (2.1)                             &    0.051 (2.2)  &    0.093   \\
                            & $[0,2]^2$ &  0.012     (2.0)                           &    0.024  (2.2) &    0.037   \\ 
                             &&&&\\
Thomas & $[0,1]^2$ &  0.0995    (3.7)                          &    $0.1418^\star$  (2.7) &    0.111   \\
                            & $[0,2]^2$ &  0.044   (4.2)                             &    0.063  (2.9) &    0.053   \\
                             &&&&\\
Variance Gamma              & $[0,1]^2$ &  0.099    (6.5)                            &    0.148  (3.8) &    0.110   \\
                            & $[0,2]^2$ &  0.050     (9.6)                           &    0.072 (5.3)  &    0.057   \\ 
                             &&&&\\

DPP        				     & $[0,1]^2$ &  NA     (3)                             &     0.1622 (3.6)  &    NA   \\ 
                            & $[0,2]^2$ &  NA   (4.1)                               &     0.1582  (5.2) &    NA \\
                            \hline
\end{tabular}
\caption{Square-root of the MISE  for different estimates of $\log
  g_0$, observation windows and models. The figures between brackets
  correspond to the average of the selected $K$s. The NA's
  are due to occurrence of non-positive estimates.  (${ }^\star$: in this setting one replication produced an outlier and is omitted in the root MISE estimation)}\label{tbl:mise}
\end{center}
\end{table}
The root MISEs are larger for the variational estimator than for the OSE and the
KDE except in the Poisson case where the KDE
has larger MISE than the VSE.
Table~\ref{tbl:mise} also reports the average of the selected $K$ for the
variational estimator and the OSE.
The averages of the selected $K$'s are pretty similar for the  Poisson and DPP models while the OSE tends to select higher $K$ than the variational method for
the Thomas and variance Gamma point processes. 


We have also compared the computing time to evaluate the OSE and VSE. The OSE is generally cheaper except when the number of
points and $R$ are large, see also the case of \emph{Capparis} in Section~\ref{sec:dataexample}.


\subsection{Data example}\label{sec:dataexample}

We apply the three estimators to the data example considered in \cite{jalilian:guan:waagepetersen:17}. That is, we consider point patterns of locations of \emph{Acalypha diversifolia} 
(528 trees), \emph{Lonchocarpus heptaphyllus} (836 trees) and 
\emph{Capparis frondosa} (3299 trees) species  in the 1995 census for 
the $1000\text{m}\times 500\text{m}$ Barro Colorado Island plot \citep{hubbell:foster:83,condit:hubbell:foster:96,condit:98}.
The intensity functions for the point patterns are estimated as in \cite{jalilian:guan:waagepetersen:17} using log-linear regression models depending on various soil and topographical variables. The estimated pair correlation functions are shown in Figure~\ref{fig:pcfsdata}.
\begin{figure}
\begin{tabular}{ccc}
\includegraphics[scale=.47]{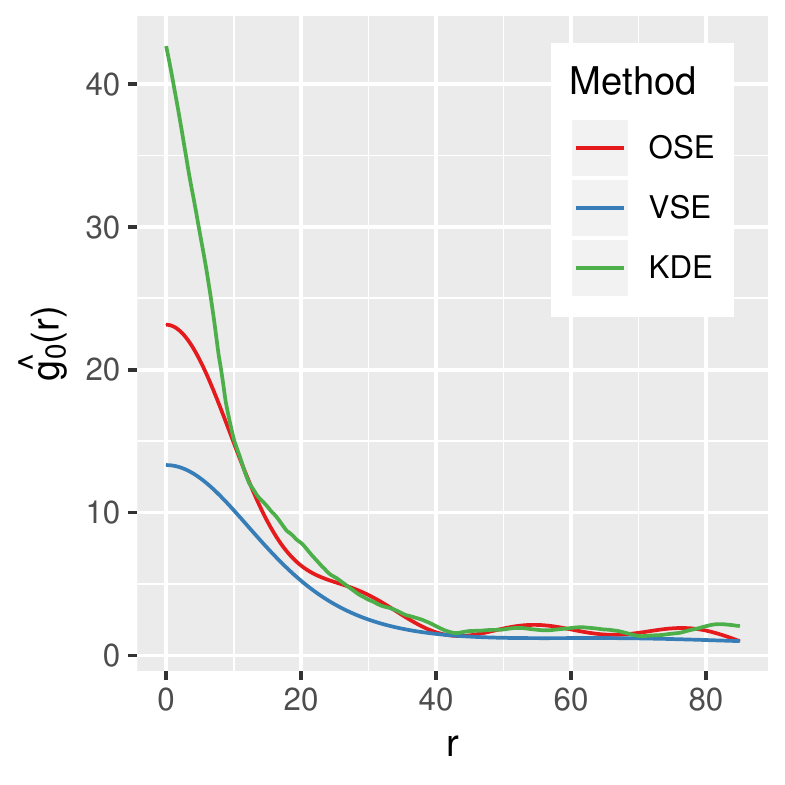} & \includegraphics[scale=.47]{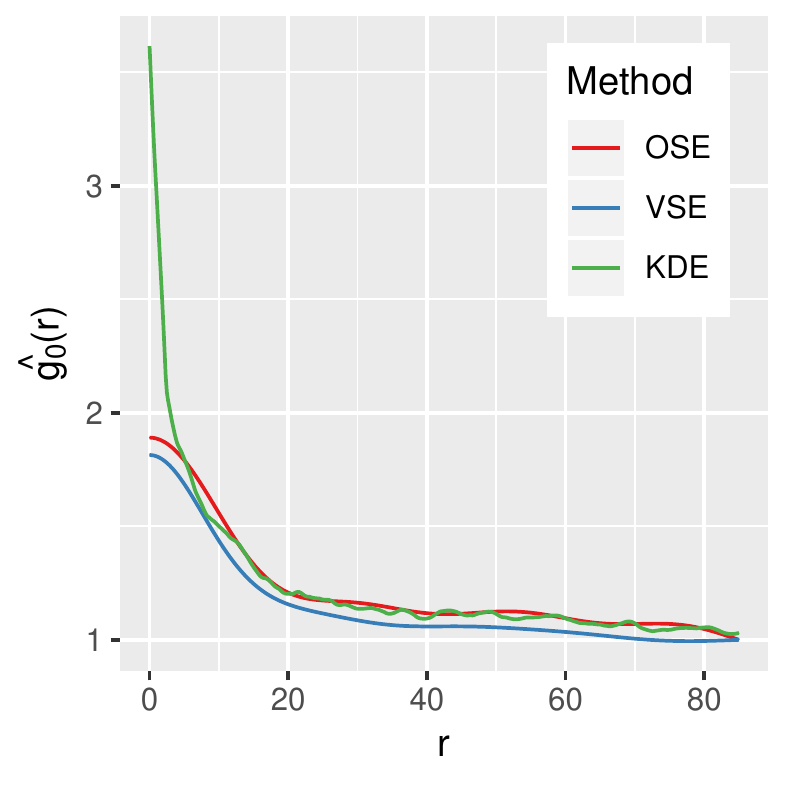} & \includegraphics[scale=.47]{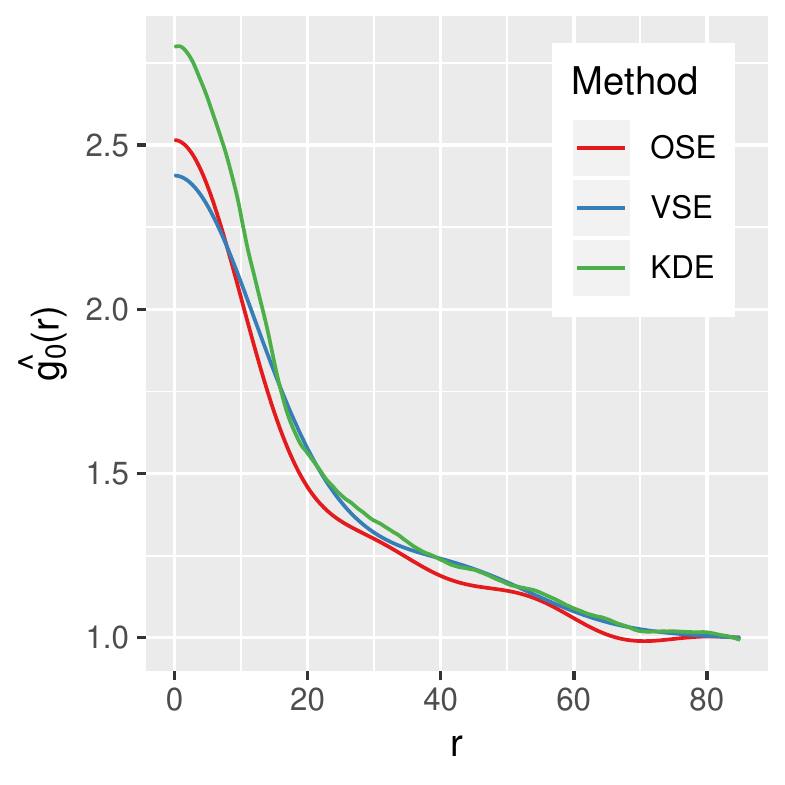}
\end{tabular}
\caption{Estimates of $g_0$ for the three species \emph{Acalypha} (left), \emph{Capparis} (middle) and \emph{Lonchocarpus} (right).}\label{fig:pcfsdata}
\end{figure}
In case of  \emph{Capparis} and \emph{Lonchocarpus}, VSE and OSE are quite similar while the kernel estimate is markedly different
from the other estimates for small lags. For \emph{Acalypha}, all estimates differ for small lags. The three estimates are very similar
for large spatial lags for all species. The selected number $K$ for the VSE are 3, 9 and 5 for \emph{Acalypha}, \emph{Capparis},
and \emph{Lonchocarpus}, while  OSE selects $K = 7$ for all species. In the case of \emph{Capparis}, the computation time (4200 seconds) is higher for the OSE than for the VSE (1244 seconds) due to the high number of points for this species.

\section{Discussion} \label{sec:discussion}

In this paper we derive variational equations based on  second order properties of a spatial point process. 
It is remarkable that in case of log-linear parametric models for the
pair correlation function, it is possible to derive variational
estimating equations which have closed form solutions for the unknown
parameters. We exploit this to construct new variational orthogonal
series type estimators for the pair correlation function. In contrast
to previous kernel and orthogonal series estimators, our new estimate
is guaranteed to be non-negative.
For large data sets,  the new estimator is further computationally faster than the previous orthogonal series estimate. However, in terms of accuracy as measured by MISE, the new estimator does not outperform the previous estimators. In the data example, the new estimator and the OSE gave similar results.

We believe there is further scope for exploring  variational equations. For instance, one could use non-orthogonal bases for expanding the log pair correlation function instead of the orthogonal Fourier-Bessel basis used in this work. One might e.g.\ in future work consider so-called frames \citep{christensen:08} or spline bases.\\[\bsl]
{\bf Acknowledgments}\\[\bsl]
Rasmus Waagepetersen's and Francisco Cuevas-Pachecho's research was supported by The Danish Council
for Independent Research | Natural Sciences, grant DFF – 7014-00074 "Statistics for point processes in space and beyond", and by the "Centre for Stochastic Geometry and 
Advanced Bioimaging", funded by grant 8721 from the Villum Foundation. The research of J.-F. Coeurjolly is funded by the Natural Sciences and Engineering Research Council of Canada.

The BCI forest dynamics research project was made possible by National
Science Foundation grants to Stephen P. Hubbell: DEB-0640386,
DEB-0425651, DEB-0346488, DEB-0129874, DEB-00753102, DEB-9909347,
DEB-9615226, \linebreak DEB-9615226, DEB-9405933, DEB-9221033, DEB-9100058,
DEB-8906869, \linebreak DEB-8605042, DEB-8206992, DEB-7922197, support from the
Center for Tropical Forest Science, the Smithsonian Tropical Research Institute, the John D. and Catherine T. MacArthur Foundation, the
Mellon Foundation, the Celera Foundation, and numerous private
individuals, and through the hard work of over 100 people from 10
countries over the past two decades. The plot project is part of the Center for Tropical Forest Science, a global network of large-scale demographic tree plots.

The BCI soils data set were collected and analyzed by J.\ Dalling,
R.\ John, K.\ Harms, R.\ Stallard and J.\ Yavitt  with support from
NSF DEB021104, 021115, 0212284, 0212818 and OISE 0314581, STRI and
CTFS. Paolo Segre and Juan Di Trani provided assistance in the
field. The covariates \texttt{dem}, \texttt{grad}, \texttt{mrvbf},
\texttt{solar} and \texttt{twi} were computed in SAGA GIS by Tomislav
Hengl \linebreak (\texttt{http://spatial-analyst.net/}).

\bibliographystyle{plainnat}
\bibliography{var}

\begin{thebibliography}{18}
\providecommand{\natexlab}[1]{#1}
\providecommand{\url}[1]{\texttt{#1}}
\expandafter\ifx\csname urlstyle\endcsname\relax
  \providecommand{\doi}[1]{doi: #1}\else
  \providecommand{\doi}{doi: \begingroup \urlstyle{rm}\Url}\fi

\bibitem[Almeida et~al.(1993)Almeida, , and Gidas]{almeida:gidas:93}
M.P. Almeida, , and B.~Gidas.
\newblock A variational method for estimating the parameters of {MRF} from
  complete or incomplete data.
\newblock \emph{Annals of Applied Probability}, 3\penalty0 (1):\penalty0
  103--136, 1993.

\bibitem[Baddeley and Dereudre(2013)]{baddeley:dereudre:13}
A.~Baddeley and D.~Dereudre.
\newblock Variational estimators for the parameters of {Gibbs} point process
  models.
\newblock \emph{Bernoulli}, 19\penalty0 (3):\penalty0 905--930, 2013.

\bibitem[Baddeley et~al.(2015)Baddeley, Rubak, and
  Turner]{baddeley:rubak:turner:15}
A.~Baddeley, E.~Rubak, and R.~Turner.
\newblock \emph{Spatial point patterns: methodology and applications with {R}}.
\newblock Chapman and Hall/CRC, 2015.

\bibitem[Baddeley et~al.(2000)Baddeley, M{\o}ller, and
  Waagepetersen]{baddeley:moeller:waagepetersen:00}
A.~J. Baddeley, J.~M{\o}ller, and R.~Waagepetersen.
\newblock Non- and semi-parametric estimation of interaction in inhomogeneous
  point patterns.
\newblock \emph{Statistica Neerlandica}, 54:\penalty0 329--350, 2000.

\bibitem[Christensen(2008)]{christensen:08}
O.~Christensen.
\newblock \emph{Frames and Bases - an introductory course}.
\newblock Applied and numerical analysis. Birkh{\"a}user, Basel, 2008.

\bibitem[Coeurjolly and M{\o}ller(2014)]{coeurjolly:moeller:14}
J.-F. Coeurjolly and J.~M{\o}ller.
\newblock Variational approach for spatial point process intensity estimation.
\newblock \emph{Bernoulli}, 20\penalty0 (3):\penalty0 1097--1125, 2014.

\bibitem[Condit(1998)]{condit:98}
R.~Condit.
\newblock \emph{Tropical Forest Census Plots}.
\newblock Springer-Verlag and R.\ G.\ Landes Company, Berlin, Germany and
  Georgetown, Texas, 1998.

\bibitem[Condit et~al.(1996)Condit, Hubbell, and
  Foster]{condit:hubbell:foster:96}
R.~Condit, S.~P. Hubbell, and R.~B. Foster.
\newblock Changes in tree species abundance in a neotropical forest: impact of
  climate change.
\newblock \emph{Journal of Tropical Ecology}, 12:\penalty0 231--256, 1996.

\bibitem[Evans and Gariepy(1992)]{evans:gariepy:92}
L.C. Evans and R.F.. Gariepy.
\newblock \emph{Measure theory and fine properties of functions.}
\newblock Studies in Advanced Mathematics. CRC Press, Boca Raton, FL, 1992.

\bibitem[Guan(2007{\natexlab{a}})]{guan:composite:07}
Y.~Guan.
\newblock A composite likelihood cross-validation approach in selecting
  bandwidth for the estimation of the pair correlation function.
\newblock \emph{Scandinavian Journal of Statistics}, 34\penalty0 (2):\penalty0
  336--346, 2007{\natexlab{a}}.

\bibitem[Guan(2007{\natexlab{b}})]{guan:leastsq:07}
Y.~Guan.
\newblock A least-squares cross-validation bandwidth selection approach in pair
  correlation function estimations.
\newblock \emph{Statistics \& Probability Letters}, 77\penalty0 (18):\penalty0
  1722--1729, 2007{\natexlab{b}}.

\bibitem[Guan et~al.(2015)Guan, Jalilian, and
  Waagepetersen]{guan:jalilian:waagepetersen:15}
Y.~Guan, A.~Jalilian, and R.~Waagepetersen.
\newblock Quasi-likelihood for spatial point processes.
\newblock \emph{Journal of the Royal Statistical Society: Series B (Statistical
  Methodology)}, 77\penalty0 (3):\penalty0 677--697, 2015.

\bibitem[Hubbell and Foster(1983)]{hubbell:foster:83}
S.~P. Hubbell and R.~B. Foster.
\newblock Diversity of canopy trees in a neotropical forest and implications
  for conservation.
\newblock In S.~L. Sutton, T.~C. Whitmore, and A.~C. Chadwick, editors,
  \emph{Tropical Rain Forest: Ecology and Management}, pages 25--41. Blackwell
  Scientific Publications, Oxford, 1983.

\bibitem[Illian et~al.(2008)Illian, Penttinen, Stoyan, and
  Stoyan]{illian:etal:08}
J.~Illian, A.~Penttinen, H.~Stoyan, and D.~Stoyan.
\newblock \emph{Statistical analysis and modelling of spatial point patterns},
  volume~70.
\newblock John Wiley \& Sons, 2008.

\bibitem[Jalilian and Waagepetersen(2018)]{jalilian:waagepetersen:18}
A.~Jalilian and R.~Waagepetersen.
\newblock Fast bandwidth selection for estimation of the pair correlation
  function.
\newblock \emph{Journal of Statistical Computation and Simulation}, 88\penalty0
  (10):\penalty0 2001--2011, 2018.

\bibitem[Jalilian et~al.(2019)Jalilian, Guan, and
  Waagepetersen]{jalilian:guan:waagepetersen:17}
A.~Jalilian, Y.~Guan, and R.~Waagepetersen.
\newblock Orthogonal series estimation of the pair correlation function of a
  spatial point process.
\newblock \emph{Statistica Sinica}, 2019.
\newblock To appear. Available at arXiv:1702.01736.

\bibitem[M\o{}ller and Waagepetersen(2004)]{moeller:waagepetersen:04}
J.~M\o{}ller and R.~Waagepetersen.
\newblock \emph{Statistical inference and simulation for spatial point
  processes}.
\newblock CRC Press, 2004.

\bibitem[Zhao(2018)]{zhao:18}
C.~Zhao.
\newblock \emph{Estimating equation estimators for the pair correlation
  function}.
\newblock Open access dissertations 2166, University of Miami, 2018.

\end{thebibliography}

\appendix

\section{Proof of Theorem~\ref{thm:VE}}  \label{sec:A}

\begin{proof}
Using  the Campbell theorem \eqref{eq:campbell} and since $\nabla \log g=(\nabla g)/g$, we start with
\begin{align*}
A :=\EE   \bigg\{ \sum_{u,v \in \bX\cap W }^{\neq} &e(u,v)   \nabla \log g(v-u) \cdot h(v-u)\bigg\} \\ 
&=\int_W \int_W \frac{1}{|W \cap W_{v-u}|} \frac{\nabla g(v-u) \cdot h(v-u)}{g(v-u)\rho(u)\rho(v)} \rho^{(2)}(u,v) \dd u \dd v \\
&=\int_W \int_W \frac{\nabla g(v-u) \cdot h(v-u)}{|W \cap W_{v-u}|}  \dd u \dd v.
\end{align*} 
Using first the invariance by translation of $h$ and $\nabla g$, second Fubini's theorem, and third a change of variables, this reduces to
\[
  A = \int_{\R^d}  \nabla g(w) \cdot h(w) \dd w.
\]
By assumption, we have using the dominated convergence theorem,
\[
  A = \lim_{n\to \infty} A_n \quad \text{ where } A_n:=\int_{B_n} \nabla g(w) \cdot h(w)\dd w.
\]
We can now use the standard trace theorem (see e.g. \cite{evans:gariepy:92}) and obtain
\[
  A_n  = -\int_{B_n} g(w) (\mathrm{div} \,h) (w)\dd w + \int_{\partial B_n} g(w) h(w) \cdot \nu(\dd w).
\]
From~\eqref{eq:condition_h}, we deduce from the dominated convergence
theorem that 
\[
  A = \lim_{n\to \infty} A_n = -\int_{\R^d} g(w) (\mathrm{div} \,h)(w)  \dd w.\\
\]
Finally, using successively a change of variable and the Campbell theorem we get
\begin{align*}
  A &= - \int_W \int_W \frac{(\mathrm{div} \,h)(v-u)}{|W\cap W_{v-u}|} \; \frac{\rho^{(2)}(u,v)}{\rho(u)\rho(v)} \dd u \dd v\\
  &= -\EE \left\{ \sum_{u,v \in \bX\cap W }^{\neq} e(u,v) \, (\mathrm{div} \,h)(v-u)\right\}
\end{align*}
which proves~\eqref{eq:VE1}. 
\end{proof}

\section{Proof of Theorem~\ref{thm:VEisotropic}} \label{sec:B}

\begin{proof}
Both \eqref{eq:VE_iso1} and \eqref{eq:VE_iso2} are proved
similarly. We focus only on~\eqref{eq:VE_iso2} and follow the proof of
Theorem~\ref{thm:VE}. Using the Campbell theorem~\eqref{eq:campbell},
the fact $(\log g_0)^\prime=g_0^\prime/g_0$ and finally a
change to polar coordinates, we have
\begin{align*}
A :=\EE   \bigg\{ &\sum_{u,v \in \bX\cap W }^{\neq} e(u,v)   (\log g_0)^\prime(\|v-u\|) h(\|v-u\|)\bigg\} \\ 
&=\int_W \int_W \frac{1}{|W \cap W_{v-u}| } \frac{g_0^\prime(\|v-u\|) h(\|v-u\|)}{g_0(\|v-u\|)\rho(u)\rho(v)} \rho^{(2)}(u,v) \dd u \dd v \\
&=\int_W \int_W \frac{g_0^\prime(\|v-u\|) h(\|v-u\|)}{|W \cap W_{v-u}| }  \dd u \dd v \\
&=  \int_{\R^d} g_0^\prime(\|w\|)h(\|w\|) \dd w \\
& =  \sad \int_0^\infty t^{d-1} g_0^\prime(t)h(t) \dd t.
\end{align*}  
Using the dominated convergence theorem, partial integration and~\eqref{eq:condition_h_iso2} we have
\begin{align*}
  \int_0^\infty t^{d-1} g_0^\prime(t)h(t) \dd t&= \lim_{n\to \infty} \int_0^n t^{d-1}g_0^\prime(t)h(t) \dd t \\
  &= - \lim_{n\to \infty} \int_0^n t^{d-1} g_0(t) \left\{ \frac{(d-1)h(t)}{t} +h^\prime(t) \right\} \dd t \\
  &= -\int_0^\infty t^{d-1} g_0(t)\left\{ \frac{(d-1)h(t)}{t} +h^\prime(t) \right\} \dd t.
\end{align*}
A change to polar coordinates and the Campbell theorem again lead to 
\begin{align*}
A &= -  \int_{R^d} g_0(\|w\|) \left\{ \frac{(d-1)h(\|w\|)}{\|w\|} +h^\prime(\|w\|) \right\} \dd w \\   
&= -\int_W \int_W \left\{ \frac{(d-1)h(\|w\|)}{\|w\|} +h^\prime(\|w\|) \right\} \frac{\rho^{(2)}(u,v)}{\rho(u)\rho(v)|W\cap W_{v-u}|} \dd u \dd v\\
&=    - \EE \left[ \sum_{u,v \in \bX \cap W }^{\neq} e(u,v) \left\{ (d-1)\frac{h(\|v-u\|)}{\|v-u\|} + h^\prime(\|v-u\|) \right\}
    \right].
\end{align*}  
\end{proof}


\end{document}